\RequirePackage[orthodox]{nag}
\PassOptionsToPackage{nameinlink}{cleveref}

\documentclass[cleveref]{lipics-v2021}
\hideLIPIcs
\nolinenumbers

\usepackage{microtype}

\usepackage{amsmath}
\usepackage{amsfonts}
\usepackage{amssymb}
\usepackage{amsthm}
\usepackage{graphicx}
\usepackage{xspace}
\usepackage[noadjust]{cite}

\title{On the approximability of graph visibility problems}

\author{Davide Bilò}{Department of Information Engineering, Computer Science, and Mathematics, University of L'Aquila, Italy}{davide.bilo@univaq.it}{https://orcid.org/0000-0003-3169-4300}{}

\author{Alessia {Di Fonso}}{Department of Information Engineering, Computer Science, and Mathematics, University of L'Aquila, Italy}{alessia.difonso@univaq.it}{https://orcid.org/0000-0002-9093-0679}{}

\author{Gabriele {Di Stefano}}{Department of Information Engineering, Computer Science, and Mathematics, University of L'Aquila, Italy}{gabriele.distefano@univaq.it}{https://orcid.org/0000-0003-4521-8356}{}

\author{Stefano Leucci}{Department of Information Engineering, Computer Science, and Mathematics, University of L'Aquila, Italy}{stefano.leucci@univaq.it}{https://orcid.org/0000-0002-8848-7006}{}

\authorrunning{D. Bilò, A. Di Fonso, G. Di Stefano, and S. Leucci} 

\Copyright{Davide Bilò, Alessia Di Fonso, Gabriele Di Stefano, and Stefano Leucci} 

\ccsdesc[500]{Mathematics of computing~Paths and connectivity problems}
\ccsdesc[500]{Mathematics of computing~Approximation algorithms}
\ccsdesc[100]{Mathematics of computing~Extremal graph theory}
\ccsdesc[100]{Mathematics of computing~Hypergraphs}

\keywords{visibility problems, mutual visibility, general position, inapproximability} 

\newcommand{\gp}{\mathit{gp}}

\newcommand{\mut}{\mu_{\rm t}}
\newcommand{\mud}{\mu_{\rm d}}
\newcommand{\muo}{\mu_{\rm o}}

\newcommand{\cliqueis}{\textsc{Max-Paired-Clique-IS}\xspace}
\newcommand{\mis}{\textsc{Max-Independent-Set}\xspace}

\begin{document}

\maketitle

\begin{abstract}
Visibility problems have been investigated for a long time under different assumptions as they pose challenging combinatorial problems and are connected to robot navigation problems. The \emph{mutual-visibility} problem in a graph $G$ of $n$ vertices asks to find the largest set of vertices $X\subseteq V(G)$, also called $\mu$-set, such that for any two vertices $u,v\in X$, there is a shortest $u,v$-path $P$ where all internal vertices of $P$ are not in $X$. This means that $u$ and $v$ are visible w.r.t. $X$. Variations of this problem are known as \emph{total}, \emph{outer}, and \emph{dual} mutual-visibility problems, depending on the visibility property of vertices inside and/or outside $X$. The mutual-visibility problem and all its variations are known to be $\mathsf{NP}$-complete on graphs of diameter $4$. 

In this paper, we design a polynomial-time algorithm that finds a $\mu$-set with size $\Omega\left( \sqrt{n/ \overline{D}} \right)$, where $\overline D$ is the average distance between any two vertices of $G$. Moreover, we show inapproximability results for all visibility problems on graphs of diameter $2$ and strengthen the inapproximability ratios for graphs of diameter $3$ or larger. More precisely, for graphs of diameter at least $3$ and for every constant $\varepsilon > 0$, we show that mutual-visibility and dual mutual-visibility problems are not approximable within a factor of $n^{1/3-\varepsilon}$, while outer and total mutual-visibility problems are not approximable within a factor of $n^{1/2 - \varepsilon}$, unless $\mathsf{P}=\mathsf{NP}$.

Furthermore we study the relationship between the mutual-visibility number and the general position number in which no three distinct vertices $u,v,w$ of $X$ belong to any shortest path of $G$. 
\end{abstract}

\section{Introduction}

Mutual-visibility problems on the Euclidean plane involve determining if a set of points or entities can see each other without any obstacles blocking their line of sight. These problems have been investigated for a long time under different assumptions. The root of visibility problems dates back to the end of the 1800s when Dudeney first introduced the famous \emph{no-three-in-line} problem in~\cite{Dudeney17}: given an $n \times n$ grid, find the maximum number of points such that there are no three points on a line and it is still an open problem.
The notion of visibility can be also defined in discrete spaces like graphs requiring that a set of entities see each other along the shortest paths connecting them without any obstacles.
Visibility problems on networks pose interesting theoretical problems both in graph theory and combinatorics but are also of practical importance in research areas like distributed computing by mobile entities in connection to robot navigation problems~\cite{bhagat20,CiceroneFSN23,CDDN23-PMC,LunaFCPSV17,PoudelAS21,SharmaVT21}. 

For robots moving in the Euclidean plane, achieving a configuration of mutual visibility is crucial. A robot whose visibility is obstructed by others, might not be able to complete its task. Conversely, when robots are mutually visible, they can all see each other and collaborate to solve problems. In communication or social networks, a subset of agents located on some nodes of the network may need to communicate in an efficient (using the shortest paths) and confidential way, that is, in such a way that the exchanged messages do not pass through other agents in the subset. 

The concept of mutual-visibility in graphs  has been recently introduced and studied in~\cite{DiStefano22}. Given a set of vertices of a graph, two vertices $u,v$ are in \emph{mutual-visibility} if there exists a shortest $u,v$-path without further vertices of the set. A set of vertices is a  \emph{mutual-visibility set} if each pair of vertices in the set is in mutual-visibility.
This graph-based mutual-visibility concept generated significant interest within the research community since its introduction producing a remarkable number of articles~\cite{EJB24, Bresar,Bujtas,CiceroneFSN23-SSS,CDDN23-PMC, CFSNP24,lagos23,variety-2023,CiceroneDK23,kuziak-2023,tian-2023+}.
In the \emph{mutual-visibility problem} the goal is to find the maximum number of vertices that can be in mutual visibility on a graph $G$.
This problem is $\mathsf{NP}$-complete~\cite{DiStefano22} on general graphs, whereas there exist exact formulas for special graph classes like paths, cycles, trees, block graphs, cographs, grids~\cite{CiceroneDK23,DiStefano22} and for both the Cartesian and the Strong product of graphs~\cite{CiceroneDK23}. 

Formally, given a connected graph $G$ and a set of vertices $X \subseteq V(G)$, two vertices $x,y\in V(G)$ are said to be $X$-\emph{visible} if there is a shortest $x,y$-path whose internal vertices do not belong to $X$. If every two vertices from $X$ are $X$-visible, then $X$ is a \emph{mutual-visibility set} (or $\mu$-set). 
Moreover, several other metrics are introduced in~\cite{variety-2023}, to consider all the possible ``visibility'' situations occurring between the vertices of a graph. 
Let $\overline{X}=V(G)\setminus X$.
A set $X$ is said to be an \emph{outer mutual-visibility set} (or $\muo$-set) if every two vertices $x,y\in X$ are $X$-visible, and every two vertices $x\in X$ and $y\in \overline{X}$ are $X$-visible.
A set $X$ is a \emph{dual mutual-visibility set} (or $\mud$-set) if every two vertices $x,y\in X$ are $X$-visible, and every two vertices $x,y\in \overline{X}$ are $X$-visible. Finally a set $X$ is said a \emph{total mutual-visibility set} (or $\mut$-set) if every two vertices $x,y\in V(G)$ are $X$-visible.
If $\tau\in \{\mu, \mud, \muo, \mut\}$, then the cardinality of the largest $\tau$-set is called the \emph{$\tau$-number} of $G$ and is denoted by $\tau(G)$. A $\tau$-set $X$ such that $|X| = \tau(G)$ is called a \emph{maximum $\tau$-set} of $G$. For each of the above variants, the respective optimization problem asks to find the maximum $\tau$-set in a given graph $G$.
\begin{figure}[t]
   \centering
   \def\svgwidth{1\columnwidth}
     \large\scalebox{1}{
\begingroup%
  \makeatletter%
  \providecommand\color[2][]{%
    \errmessage{(Inkscape) Color is used for the text in Inkscape, but the package 'color.sty' is not loaded}%
    \renewcommand\color[2][]{}%
  }%
  \providecommand\transparent[1]{%
    \errmessage{(Inkscape) Transparency is used (non-zero) for the text in Inkscape, but the package 'transparent.sty' is not loaded}%
    \renewcommand\transparent[1]{}%
  }%
  \providecommand\rotatebox[2]{#2}%
  \newcommand*\fsize{\dimexpr\f@size pt\relax}%
  \newcommand*\lineheight[1]{\fontsize{\fsize}{#1\fsize}\selectfont}%
  \ifx\svgwidth\undefined%
    \setlength{\unitlength}{375.36852679bp}%
    \ifx\svgscale\undefined%
      \relax%
    \else%
      \setlength{\unitlength}{\unitlength * \real{\svgscale}}%
    \fi%
  \else%
    \setlength{\unitlength}{\svgwidth}%
  \fi%
  \global\let\svgwidth\undefined%
  \global\let\svgscale\undefined%
  \makeatother%
  \begin{picture}(1,0.26072645)%
    \lineheight{1}%
    \setlength\tabcolsep{0pt}%
    \put(0,0){\includegraphics[width=\unitlength,page=1]{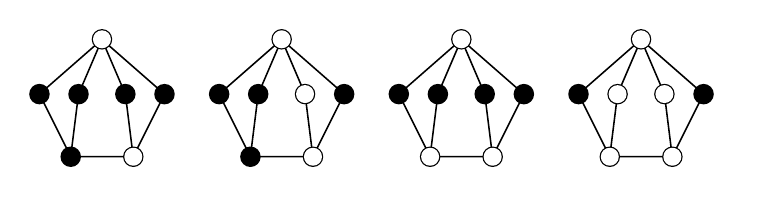}}%
    \put(0.08098319,0.0169032){\color[rgb]{0,0,0}\makebox(0,0)[lt]{\lineheight{1.25}\smash{\begin{tabular}[t]{l}max $\mu$-set\end{tabular}}}}%
    \put(0.30875947,0.0169032){\color[rgb]{0,0,0}\makebox(0,0)[lt]{\lineheight{1.25}\smash{\begin{tabular}[t]{l}max $\mud$-set\end{tabular}}}}%
    \put(0.54053189,0.0169032){\color[rgb]{0,0,0}\makebox(0,0)[lt]{\lineheight{1.25}\smash{\begin{tabular}[t]{l}max $\muo$-set\end{tabular}}}}%
    \put(0.76830826,0.0169032){\color[rgb]{0,0,0}\makebox(0,0)[lt]{\lineheight{1.25}\smash{\begin{tabular}[t]{l}max $\mut$-set\end{tabular}}}}%
  \end{picture}%
\endgroup%
}
    \caption{Examples of maximum $\tau$-sets on a graph $G$. Note that the $\mud$-set is neither a $\mut$-set nor a $\muo$-set. The $\mud$-set is neither a $\mut$-set nor a $\muo$-set. The $\mut$-set is a feasible $\mu$-set, $\mud$-set and $\muo$-set that is not maximum.}
\label{fig:mu-sets}
\end{figure}
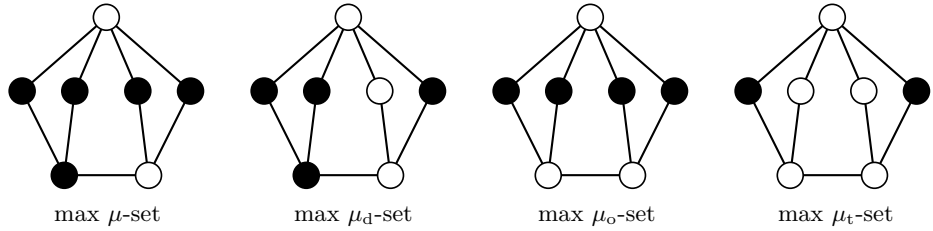

\smallskip

The mutual-visibility problem is connected to classical topics in combinatorics. For example, solving this problem in the Cartesian product of complete graphs is equivalent to solving an instance of Zarankiewicz's problem (see~\cite{CiceroneDK23}). Finding the smallest maximal mutual-visibility set of a graph has been proven to be closely related to a classical Bollob\'as-Wessel theorem (see~\cite{Bresar}). Also, the optimization visibility problems introduced above can be reformulated as Tur\'an-type problems on hypergraphs and line graphs (see~\cite{Bujtas,EJB24}).
The mutual-visibility problem is also related to the \emph{general position problem} in graphs that asks to determine a largest set $X$ of vertices of $G$ such that no three vertices of $X$ lie on a common shortest path~\cite{ManuelK18,manuel-2018, ullas-2016}. Such a set is called \emph{general position set} (or $gp$-set) and the cardinality of a maximum $gp$-set is called the $gp$-number of $G$ and denoted by $gp(G)$.

\subparagraph*{Our results.} For a given graph $G$ of $n$ vertices, the problem of finding the $\gp$-number $\gp(G)$, as well as the problems of finding the $\tau$-number $\tau(G)$, for $\tau \in \{\mu, \mud, \muo, \mut\}$, have been proved to be $\mathsf{NP}$-complete, respectively in~\cite{ManuelK18} and in~\cite{variety-2023}. All the $\mathsf{NP}$-completeness results hold for graphs of diameter 4 or larger. Then many other works restricted their attention to the study of these quantities for special classes of graphs (e.g., see~
\cite{AnandCCKT19,EJB24,CiceroneDK23,ManuelK18,lagos23,CFSNP24,klavzar-2021,tian-2021,Bujtas,kuziak-2023,tian-2023+}). However, these problems have not yet been considered from the approximability point of view. 
In this paper, we provide an  algorithm  that finds a $\mu$-set with size $\Omega\left( \sqrt{n/ \overline{D}} \right)$, where $\overline D$ is the average distance between any two vertices of the graph, and we present strong inapproximability results about the computation of $\gp(G)$, $\mu(G)$, $\mud(G)$, $\muo(G)$, and $\mut(G)$ on graphs of diameter $2$ or larger, as summarized in \Cref{tab:results}.  We also study the relationship between the general position number and the mutual-visibility number in graphs of diameter~$2$.

\subparagraph*{Structure of the paper.} 
The next section provides some preliminary notions.  In \Cref{sec:Algo} we present the algorithm for approximating a maximum mutual-visibility set. Then, we study the computational complexity of finding maximum $\tau$-sets, for $\tau \in \{\mu, \mud, \muo, \mut\}$ in \Cref{sec:InapVis,sec:inapprox3}. In \Cref{sec:inaprox_gp} we show the inapproximability result of the general position problem and the relationship between the general position number and the mutual-visibility number. Finally, \Cref{sec:open} discusses some open problems.

\begin{table}[t]
\caption{Summary of the results in this paper. Here $\varepsilon$ denotes a positive constant of choice.}\label{tab:results}
\centering
\begin{tabular}{|c|c|c|c|}
\hline
Measure(s) &  Result        & Notes    & Reference                                  \\ \hline                $\mu$ &  $\mu$-set of size $\Omega(\sqrt{n / \overline{D} })$ in poly time & $\overline{D}$ is the avg. distance in $G$ & \Cref{theo:approx} \\ \hline       $\mu, \mud, \muo, \mut$  & $\mathsf{APX}$-Hard   & $ \text{diam}(G) = 2$   & \Cref{tau-apx}   \\ \hline     
$\mut$  & Not approximable within $n^{1/3-\varepsilon}$ & $\text{diam}(G) = 2$ & \Cref{mut-diam2}           \\ \hline
$\mu$, $\mud$  & Not approximable within $n^{1/3-\varepsilon}$ & $\text{diam}(G) \geq 3$ & \Cref{theo:inapx_diam_3} \\ \hline
$\muo$, $\mut$ & Not approximable within $n^{1/2-\varepsilon}$ & $\text{diam}(G) \geq 3$ & \Cref{theo:inapx_diam_3} \\ \hline
$\gp$    & Not approximable within $n^{1-\varepsilon}$   & $ \text{diam}(G) = 2$ & \Cref{thm:gp_inapx}     \\ \hline

\end{tabular}
\end{table}

\section{Preliminaries}
We consider undirected graphs and unless otherwise stated, all graphs in the paper are connected. Given a graph $G$, $V(G)$ and $E(G)$ are used to denote its vertex set and its edge set, respectively. The order of $G$, that is $|V(G)|$, is denoted by $n(G)$, and its size $|E(G)|$ is denoted by $m(G)$. We remove the argument $G$ when it is clear from the context. If $X\subseteq V(G)$, then $G[X]$ denotes the subgraph of $G$ induced by $X$. 
For a natural number $k$, we set $[k] = \{1,\ldots, k\}$.

An \emph{independent set} is a set of vertices of $G$, no two of which are adjacent. The cardinality of a largest independent set is the {\em independence number} $\alpha(G)$ of $G$.

 The \emph{complete graph} (or \emph{clique}) $K_n$, $n\ge 1$, is the graph with $n$ vertices where each pair of distinct vertices are adjacent. A \emph{subcubic graph} is a graph where each vertex has degree at most $3$.

The distance between two vertices $u,v$ in a graph $G$ is denoted $d(u,v)$ and is the number of edges in a shortest $u,v$-path. The {\em diameter} of $G$ is the maximum distance between pairs of vertices of the graph.  

The \emph{Cartesian product} $G\times H$  of graphs $G$ and $H$ both have the vertex set $V(G)\times V(H)$. In $G\times H$, vertices $(g,h)$ and $(g',h')$ are adjacent if either $g=g'$ and $hh'\in E(H)$, or $h=h'$ and $gg'\in E(G)$. 
A {\em layer} in $G\times H$ is a subgraph induced by the vertices in which one of the coordinates is fixed. Note that each layer is isomorphic either to $G$ or $H$. 

If $G$ is a graph then, by definition $\mu(G) \ge \muo(G) \ge \mut(G)$ and $\mu(G) \ge \mud(G) \ge \mut(G)$.

\section{A polynomial-time algorithm for computing mutual-visibility sets}
\label{sec:Algo}

We consider graphs $G$ with $n \ge 7$ vertices as for graphs with at most $6$ vertices we can compute a $\mu$-set of maximum size in a brute-force manner. We denote by $\binom{V(G)}{2}$ the set of all unordered pairs of distinct vertices in $V(G)$.
For any $\{u,v\} \in \binom{V(G)}{2}$, consider a shortest path $\langle u, x_1, x_2, \dots, x_k, v \rangle$ from $u$ to $v$ in $G$ and let $B( \{ u,v \} ) = \{x_1, \dots, x_k \}$. Notice that it might be $k=0$, in which case $B(\{u,v\}) = \emptyset$.

We build a $3$-uniform hypergraph $H$ (i.e., a hypergraph in which each hyperedge contains exactly $3$ vertices) as follows: the set of vertices of $H$ is $V(G)$ and there exists a hyperedge $\{u, v, x\}$ for each $\{u,v\} \in \binom{V(G)}{2}$ and $x \in B( \{ u,v \} )$. See Figure~\ref{fig:hypergraph} for an example.

\begin{figure}[t]
   \centering
    \includegraphics[]{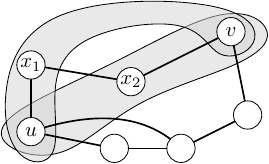}
    \caption{A sample graph $G$ and the hyperedges (shown as sets of vertices) added to the hypergraph $H$ for the pair of vertices $u,v$ when $\langle u, x_1, x_2, v\rangle$ is chosen as a shortest path.}
    \label{fig:hypergraph}
\end{figure}

An independent set of a hypergraph $H$ is a subset $S$ of vertices $V(G)$ such that, there exists no hyperedge $e$ such that $e \subseteq S$. 

Observe that an independent set $S$ of $H$ is a $\mu$-set of $G$. Indeed, for every two distinct vertices $u$ and $v$ such that $u,v \in S$, it must be the case that no vertex $x \in B(\{u,v\})$ is in $S$ as $H$ contains the hyperedge $\{u,x,v\}$.

We compute an independent set $S$ of $H$ using the algorithm of~\cite{CaroT91}. 
For an integer $\ell \ge 0$ and a real $r \ge 0$, we define $\binom{r}{\ell} = \frac{1}{\ell!} \prod_{i=0}^{\ell-1} (r-i)$.
It is known~\cite{ShachnaiS01} that $|S|$ can be lower bounded as:
\[
|S| \ge \sum_{v \in V(G)} \frac{1}{\binom{\delta(v) + \frac{1}{2}}{\delta(v)}} = \Theta\left( \sum_{v \in V(G)} \frac{1}{\sqrt{ 1 + \delta(v) }} \right),
\]
where $\delta(v)$ denotes the number of hyperedges incident to $v$ and we used $ \binom{\delta(v) + \frac{1}{2}}{\delta(v)} = \Theta( \sqrt{ 1 + \delta(v) } )$.

Let $m(H)$ be the number of hyperedges in $H$, and let $\overline{D} = \frac{2}{n(n-1)} \sum_{ \{u,v\} \in \binom{V(G)}{2}} d(u,v)$ be the average distance in $G$. Using the fact that $H$ is $3$-uniform we have:

\begin{align*}
    \sum_{v} \delta(v) &= \frac{1}{3} m(H) \le 
    \frac{1}{3} \sum_{\{u,v\} \in \binom{V(G)}{2}} |B( \{u,v\} )|=
    \frac{1}{3} \sum_{ \{u,v\} \in \binom{V(G)}{2}} \left(  d(u,v) - 1 \right)
    \\
    &= \frac{n(n-1)}{6} ( \overline{D} -1)
    \le \frac{n^2}{6} \overline{D} - n.
\end{align*}

Let $\varphi(x) = \frac{1}{\sqrt{1+x}}$ and notice that $\varphi$ is convex and monotonically decreasing for $x \ge 0$. By Jensen's inequality:
\[
    \sum_{v \in V(G)} \frac{1}{\sqrt{ 1 + \delta(v) }} = \sum_{v \in V(G)} \varphi(\delta(v)) \ge n \cdot \varphi\left( \frac{1}{n}  \cdot \sum_{v \in V(G)} \delta(v) \right)
    \ge n \cdot \varphi\left( \frac{n}{6} \overline{D} - 1 \right)
    = \sqrt{\frac{6n}{\overline{D}}},
\]
hence $|S| = \Omega\left(\sqrt{\frac{n}{\overline{D}}}\right)$. We have therefore shown:

\begin{theorem}\label{theo:approx}
    Given an input graph $G$ on $n$ vertices, it is possible to find, in polynomial time, a $\mu$-set of $G$ having size $\Omega\left( \sqrt{n/ \overline{D}} \right)$, where $\overline{D} = \frac{2}{n(n-1)} \sum_{ \{u,v\} \in \binom{V(G)}{2} } d(u,v)$ is the average distance in $G$.
\end{theorem}

\section{Inapproximability of visibility problems on graphs of diameter 2}\label{sec:InapVis}

In this section we show inapproximability results for the problem of computing $\tau$-set of maximum size, where $\tau \in \{\mu, \mud, \muo, \mut\}$, 
for graphs with diameter $2$ via reductions from the \mis problem.

Given an undirected graph $H$, the \mis problem asks to compute an independent set of $H$ of maximum cardinality.

    The \mis problem on subcubic graphs is not approximable in polynomial time within a factor of $c'$ for a suitable constant $c' > 1$, unless $\mathsf{P}=\mathsf{NP}$, see~\cite{hardnessCubic}. Moreover, for general graphs, it cannot be approximated  within a factor of $O(n(H)^{1-\delta})$ in polynomial time, for any constant $\delta>0$, see~\cite{Zuckerman07}.
    
    The reduction we are going to show is from the \mis problem on connected graphs.
    Given a connected graph $H$ with $n(H) \ge 3$ vertices  which is the input instance of the \mis problem, we construct the graph $G$ as a function of $H$ and an additional integer parameter $L \ge 1$. 
    The set of vertices of $G$ consists of  the union of (i) all edges in $E(H)$, (ii) $L$ copies $v_1, \dots, v_L$ of each vertex $v \in V(H)$, and (iii) two new vertices $y$ and $z$.
    The set of edges of $G$ contains (i) all edges $(u_1, v_1), \dots, (u_L, v_L)$ for each pair of distinct vertices $u,v \in V(H)$  (ii) all edges $(e, v_1), \dots, (e, v_L)$ for each $e \in E(H)$ and each $v \in V(H)$ such that $v$ is an endvertex of $e$, (iii) all edges $(y, v_1), \dots, (y, v_L)$ for all $v \in V(H)$, and finally (iv) all edges $(z, e)$ for $e \in E(H)$.
Observe that $G$ has diameter $2$. \Cref{fig:H2G} shows the construction of a graph $G$ corresponding to a particular graph $H$.

\begin{figure}[t]
   \centering
   \def\svgwidth{1\columnwidth}
     \large\scalebox{1}{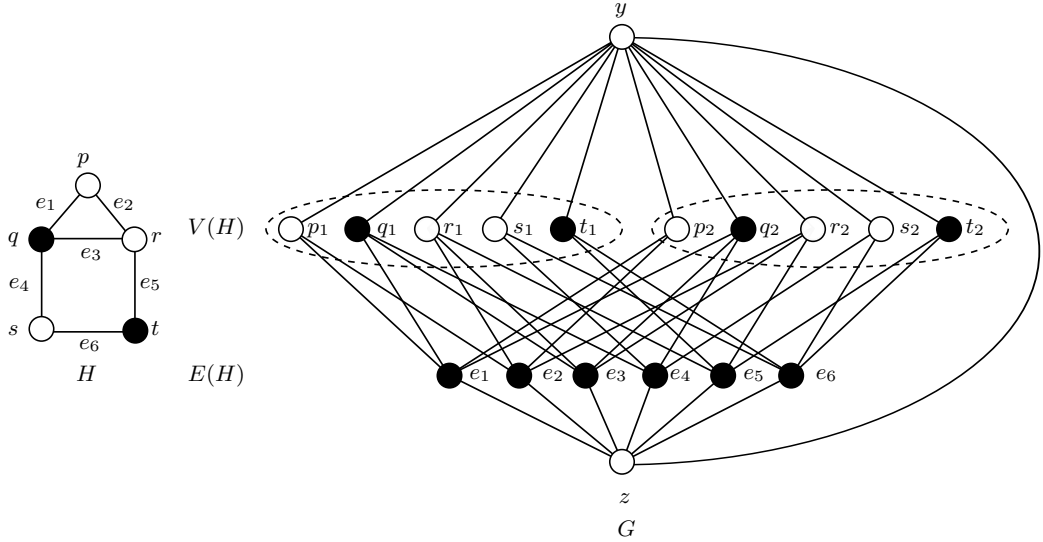}
    \caption{Construction of graph $G$ from graph $H$; $L=2$. Vertices inside the dashed ellipses induce complete graphs.}
    \label{fig:H2G}
\end{figure}

We first present some technical lemmas that will be instrumental to proving our  \linebreak inapproximability results.

    \begin{lemma}
        \label{label:mu-set-to-is}
        Consider the graph $G$ obtained from $H$ when $L=1$.
        Given a $\mu$-set $M$ of $G$, it is possible to compute, in polynomial time, an independent set $S$ of $H$ of size at least $|M| - m(H) - 4$.        
    \end{lemma}
    \begin{proof}
        Let $S'$ be the set of all vertices $v \in V(H)$ such that $v_1 \in M$.
        If $|S'| \le 2$ we choose $S = \emptyset$, and the claim follows from $|M| \le m(H) + |S'| + 2 \le m(H) + 4$, i.e., $|S| = 0 \ge |M| - m(H) - 4$. Hence, in the rest of the proof we only consider the case $|S'| \ge 3$.

        Choose an arbitrary end vertex $r_e$ of each edge $e = (u,v) \in E(H)$ such that $e \not\in M$ and define 
        $R = \{ r_e \mid e \in V(H) \setminus M \}$.
        We choose $S = S' \setminus R$.

        To see that $S$ is an independent set of $H$, consider an arbitrary edge $e = (u,v) \in E(H)$, let $w \in S' \setminus \{u,v\}$
        and notice that it cannot be the case that $u,v \in S$ since this would imply $e \in M$ and $e$ would not be in mutual visibility with $w_1$ in $G$. Indeed, the only two shortest paths connecting $w_1$ to $e$ are $\langle w_1,u,e\rangle$ and $\langle w_1,v,e\rangle$.

        Using $|S'| + |M \cap E(H)| \ge |M| - 2$, we can write:
        \begin{align*}
            |S| & \ge |S'| - |R| \ge |S'| - | E(H) \setminus M | 
            = |S'| - | E(H) \setminus (M \cap E(H)) | \\
            &= |S'| - ( m(H)  -  |M \cap E(H)| ) 
            =  |S'| + |M \cap E(H)| - m(H)
            \ge |M| - 2 - m(H). \tag*{\qedhere}
        \end{align*}
    \end{proof}

    \begin{lemma}
        \label{lemma:total_set_to_is}
        Given a $\mut$-set $X$ of $G$, it is possible to compute, in polynomial time, an independent set $S$ of $H$ of size at least $\frac{|X| - m(H) - 2}{L}$.       
    \end{lemma}
    \begin{proof}
        Consider any $e=(u,v) \in E(H)$, and let $w \in V(H) \setminus \{u,v\}$. For any $i \in [L]$, the only two shortest paths between $w_i$ and $e$ in $G$ are $\langle w_i, u_i, e \rangle$ and $\langle w_i, v_i, e \rangle$, hence at least one of $u_i$ and $v_i$ is not in $X$.

        Let $S_i$ be the set containing all vertices $v \in V(H)$ such that $v_i \in X$. By the above discussion we have that each $S_i$, with $i \in [L]$, is an independent set of $H$ and
        we choose $S$ as any of the sets $S_i$ of maximum cardinality.
        Since $\sum_{i=1}^L |S_i| \ge |X| - m(H) - 2$, we have:
        \begin{equation*}            
           |S| = \max_{i\in[L]} |S_i| \ge \frac{1}{L} \sum_{i=1}^L |S_i| \ge \frac{|X| - m(H) - 2}{L}. \tag*{\qedhere}
        \end{equation*}
    \end{proof}

    \begin{lemma}
        \label{lemma:is_to_total}
        $\mu_t(G) \ge L \cdot \alpha(H) + m(H)$.
    \end{lemma}
    \begin{proof}
        Let $S$ be a maximum independent set of $H$.
        We define $X$ as the subset of vertices of $G$ that contains all $e \in E(H)$ plus all copies $v_1, \dots, v_L$ for each $v \in S$.
        Clearly, $|X| = L \cdot |S| + m(H) = L \cdot \alpha(H) + m(H)$.
        We now argue that $X$ is a $\mut$-set of $G$.

        Each vertex $v_i \in V(G)$ for $v \in V(H)$ and $i \in [L]$ is in $X$-visibility with $y$ (via the edge $(u,y)$), with $z$ (via the shortest path $\langle u, y, z \rangle)$, with any other vertex $u_i$ for $u \in V(H)$ (via the shortest path consisting of the sole edge $(v_i,u_i)$), and with any other vertex $u_j$ with $u \in V(H)$ and $j \in [L] \setminus \{i\}$ (via the shortest path $\langle v_i, y, u_j \rangle$).
    
        Similarly, each vertex $e \in E(H)$ is in $X$-visibility with $z$ (via the edge $(e, z)$), with $y$ (via the shortest path $\langle e, z, y \rangle)$ and with any other vertex $f \in E(H)$ (via the shortest path $\langle e, z, f \rangle$).

        It only remains to argue that each vertex $x_i$ with $x \in V(H)$ and $i \in [L]$ is in $X$-visibility with any other vertex $e = (u,v) \in E(H)$. 
        If $x \in \{u, v\}$ this is clearly the case, since $G$ contains the edge $(x_i,e)$. If $x \not\in \{u, v\}$ then let $w$ be an arbitrary vertex in $\{ u, v \} \setminus S$ (such a $w$ exists since $S$ is an independent set of $H$) and observe that $\langle x_i, w_i, e \rangle$ is a shortest path between $x_i$ and $e$ in $G$ with $w_i \not\in X$. 
    \end{proof}
    We can now prove the two main inapproximability results of this section.
    \begin{theorem}\label{tau-apx}
        For every $\tau \in \{\mu, \mud, \muo, \mut\}$, the problem of computing a maximum-size $\tau$-set of an input graph $G$ with diameter $2$ is $\mathsf{APX}$-Hard.
    \end{theorem}
    \begin{proof}
        Let $c \in \left(1, 1+ \frac{1}{12} \right]$ be a constant whose exact value will be given later.
        We show how to transform any $c$-approximation algorithm $A$ for the problem of computing a $\tau$-set into a $3c$-approximation algorithm for the minimum independent set problem on subcubic graphs. 

        Given an instance $H$ of the \mis problem, where $H$ is a subcubic graph, we can assume w.l.o.g.\ that $H$ is connected (otherwise we can apply the following arguments on each connected component of $H$), that $m(H) \ge n(H)$ (otherwise $H$ is a tree and we can find a maximum independent set in polynomial time), and that $n(H) \ge 104$ (otherwise we can find a maximum independent set in constant time by brute force).
        
        Since $H$ is subcubic, we must have $\alpha(H) \ge \frac{n(H)}{4}$.\footnote{A independent set $S$ of  $H$ having size at least $\frac{n(H)}{4}$ can be computed by the greedy algorithm that starts from $S = \emptyset$ and iteratively (i) adds an arbitrary vertex $v$ to $S$, and (ii) deletes $v$ and all its (at most $3$) neighbors form $H$ until no vertices are left.}
        We apply our reduction to $H$ with $L=1$ to obtain a graph $G$.
        From \Cref{lemma:is_to_total} we know that an optimal $\tau$-set for $G$ has size at least $\alpha(H) + m(H)$, and hence the $c$-approximate solution $M$ computed by running algorithm $A$ on $G$ has size at least $\frac{\alpha(H) + m(H)}{c}$.
        Since $M$ is always a mutual-visibility set, \Cref{label:mu-set-to-is} allows us to compute, in polynomial time, an independent set $S$ of $H$ that satisfies:
        \[
                |S| \ge |M| - m(H) - 4 
                \ge \frac{\alpha(H) + m(H)}{c} - m(H) - 4
                = \frac{\alpha(H) - (c-1)m(H) - 4c}{c}.
        \]

        Such a set $S$ is exactly the one returned by our approximation algorithm for the independent set problem on subcubic graphs. The achieved approximation ratio is:
        \begin{multline*}
            \frac{\alpha(H)}{|S|} \le  \frac{c \alpha(H)}{\alpha(H) - (c-1) m(H) - 4c}
            = c \left(1 + \frac{(c-1)m(H)+4c}{\alpha(H) - (c-1)m(H)-4c}\right) \\
             \le c \left(1 + \frac{(c-1) \frac{3 n(H)}{2} +4c}{\frac{n(H)}{4} - (c-1)\frac{3n(H)}{2}-4c}\right)
            \le c \left(1 + \frac{n(H) + \frac{104}{3}}{n(H) - \frac{104}{3}}\right) 
            \le 3c,
        \end{multline*}
        where we used $m(H) \le \frac{3n(H)}{2}$ (since $H$ is subcubic), $c \le \frac{13}{12}$, and the fact that $\frac{x + \frac{104}{3}}{x - \frac{104}{3}} \le 2$ for all $x \ge 104$.

        As the \mis problem is not approximable in polynomial time within a factor of $c'$ for a suitable constant $c' > 1$, unless $\mathsf{P}=\mathsf{NP}$, the claim follows by choosing $c = \min\left\{ 1 + \frac{1}{12},  \frac{c'}{3} \right\}$. 
    \end{proof}

    \begin{theorem}\label{mut-diam2}
        The problem of computing a maximum-size $\mut$-set of an input graph $G$ having diameter $2$ is not approximable within $n^{\frac{1}{3}-\epsilon}$, for any constant $\varepsilon>0$, unless $\mathsf{P} = \mathsf{NP}$.
    \end{theorem}
    \begin{proof}
        Let $c>0$ be a constant of choice. We show how to transform any $n(G)^c$-approximation algorithm $A$ for the problem of computing a $\mut$-set of a graph $G$ with $n(G)$ vertices into a $O(n(H)^{3c})$-approximation algorithm for the problem of computing minimum independent of a graph $H$ with $n(H)$ vertices.

        Consider an instance $H$ of the \mis problem, let $n(H)$ (resp.\ $m(H)$) be the number of vertices (resp.\ edges) of $H$, assume w.l.og.\ that $n(H) \ge 3$ (otherwise a maximum independent set of $H$ can be found in constant time by brute force), and apply our reduction with $L=n(H)^2$ to construct $G$.
        Note that $n(G) = L \cdot n(H) + m(H) + 2 = \Theta(n(H)^3)$.

        From \Cref{lemma:is_to_total}, we have $\mut(G) \ge n(H)^2 \alpha(H) + m(H)$, hence the $\mut$-set $X$ computed by running algorithm $A$ on $G$ has size at least $\frac{n(H)^2 \alpha(H) + m(H)}{n(G)^c}$.

        As shown by \Cref{lemma:total_set_to_is}, we can convert $X$ into an independent set $S$ of $H$ having size $\frac{|X|- m(H)-2}{n(H)^2}$. Then:
        \begin{align*}
            |S| &\ge  \frac{|X|- m(H)-2}{n(H)^2}
            \ge \frac{\frac{n(H)^2 \alpha(H)}{n(G)^c} - m(H) - 2}{n(H)^2}
            = \frac{\alpha(H)}{n(G)^c} - \frac{m(H) + 2}{n(H)^2} \\
            &\ge \frac{\alpha(H)}{ \Theta( n(H)^{3c} )} - 1
            = O\left( \frac{\alpha(H)}{n(H)^{3c}} \right).
        \end{align*}

        Hence, for any constant $\varepsilon>0$, no polynomial-time $n(G)^{\frac{1}{3}-\epsilon}$-approximation algorithm can exist for the problem of computing a maximum $\mut$-set, unless  $\mathsf{P} = \mathsf{NP}$, since it would imply the existence of a polynomial-time $O(n(H)^{1-\varepsilon/3})$-approximation algorithm for the \mis problem. 
        \end{proof}

\section{Inapproximability of visibility problems on graphs of diameter 3}\label{sec:inapprox3}

In this section we show stronger inapproximability results of visibility problems for graphs of diameter 3. More precisely, given a graph $G$ of $n$ vertices and diameter of at least 3, we show that, for every $\tau \in \{\mu,\mud,\muo,\mut\}$ and every constant $\varepsilon >0$, it is not possible to design a polynomial-time algorithm that computes a $\tau$-set whose size approximates the value of $\tau(G)$ within a factor of $n^{1/(\beta+1)-\varepsilon}$, where $\beta=2$ if $\tau \in \{\mu,\mud\}$ and $\beta=1$ if $\tau \in \{\muo,\mut\}$, unless $\mathsf{P} = \mathsf{NP}$. 

Given a graph $H$ of $N$ vertices and a clique $K_L$ of $L\geq 1$ vertices, consider the graph $G=K_L \times H$. 
We denote by $H_i$ the layer consisting in the $i$-th copy of $H$ in $K_L \times H$ and we denote by $v_i$ the copy of $v \in V(H)$ that belongs to the layer $H_i$.

We prove useful connections between $\tau$-sets in $K_L \times H$ and independent sets in $H$.

\begin{lemma}\label{lemma:property_cartesian_product_clique_graph}
Given a $\tau$-set $X$ for $K_L \times H$, with $\tau \in\{\mu,\mud,\muo,\mut\}$, where $H$ is a graph with $N$ vertices and $L\geq 1$, we can find in polynomial time a subset $X'\subseteq X$ that satisfies the following two conditions:
\begin{itemize}
    \item for every $1\in [L]$, $X' \cap V(H_i)$ is an independent set of $H_i$;
    \item $|X'|\geq |X|-N^\beta$, where $\beta=2$ if $\tau \in\{\mu, \mud\}$ and $\beta = 1$ if $\tau\in\{\muo, \mut\}$.
\end{itemize} 
\end{lemma}
\begin{proof}
We say that an edge $(u,v) \in E(H)$ \emph{appears $k$ times} in $K_L \times H$ w.r.t. $X$ if there are $k$ distinct copies $(u_{i_1},v_{i_1}), \ldots, (u_{i_k},v_{i_k})$ of the edge $(u,v)$ such that $u_{i_j},v_{i_j} \in X$ for every $j\in [k]$.

We show that no edge $(u,v) \in E(H)$ appears $k\geq 2$ times in $K_L \times H$ w.r.t. $X$. Indeed, if there were two distinct copies $(u_i,v_i)$ and $(u_j,v_j)$ of $(u,v)$ in $H$ such that $u_i,v_i,u_j,v_j \in X$, then $u_i$ and $v_j$ would not be $X$-visible as the only two shortest paths from $u_i$ to $v_j$ pass through vertices $u_j$ and $v_i$, respectively. Furthermore, when $\tau \in \{\muo,\mut\}$, for each edge $(u,v) \in E(H)$ that appears $k=1$ times in $K_L \times H$ w.r.t. $X$, i.e., such that $u_i,v_i \in X$ for some $i \in [L]$, we have that $u_j,v_j \not \in X$ for every $j\in [L]$, with $i\neq j$. Indeed, if w.l.o.g. $u_j$ were contained in $X$, $u_i$ and $v_j$ would not be $X$-visible.

We compute a set $X''$ in polynomial time as follows. We start with $X''=\emptyset$. Next, for each edge $(u,v) \in E(H)$  that appears $k=1$ times in $K_L \times H$ w.r.t. $X$, i.e., there is a exactly one copy $(u_i,v_i)$ of $(u,v)$ in $H_i$ such that $u_i,v_i \in X$, we add both $u_i$ and $v_i$ to $X''$.
By construction, $X''$ is a subset of $X$ of size $|X''| \leq N^2$ as $H$ contains at most $\binom{N}{2}$ edges, each of which contributes with at most 2 vertices in $X''$. The upper bound on the size of $X''$ can be refined to $|X''| \leq N$ when $\tau \in \{\muo,\mut\}$ for the following reason.
Adjacent edges in $H$ that appear $k=1$ times in $K_L \times H$ w.r.t. $X$ must all appear in the same copy of $H$ in $K_L \times H$. This implies that each vertex $u$ of $H$ contributes with at most 1 vertex in $X''$.

Let $X'=X\setminus X''$. By construction, $|X'|\geq |X|-N^2$ when $\tau \in \{\mu,\mud\}$ and $|X'|\geq |X|-N$ when $\tau \in \{\muo,\mut\}$. Moreover, each edge of $H$ does not appear in $K_L \times H$ w.r.t. $X'$. As a consequence, $X' \cap V(H_i)$ is an independent set of $H_i$, for every $i \in [L]$. 
\end{proof}

\begin{lemma}\label{lemma:is_graph_feasible_clique_graph}
Let $H$ be a graph with $N\geq 2$ vertices containing a  vertex $z$ of degree $N-1$ and let  $S \subseteq V(H)$ be an independent set of $H$ that does not contain $z$. Let $L$ be a positive integer. The graph $K_L \times H$ has diameter of at most 3. Moreover, the set $X=\cup_{i \in [L]}S_i$, where $S_i=\{v_i \mid v \in S\}$ is the copy of $S$ in $H_i$, is a $\mut$-set of $K_L \times H$.
\end{lemma}
\begin{proof}
The graph $H$ has diameter of at most $2$ as $z$ is adjacent  to all other vertices of $H$. As a consequence, the graph $K_L\times H$ has diameter of at most 3 as all $L$ copies of $z$ in $K_L \times H$ form a clique. 

By definition, the endvertices of any edge of $K_L \times H$ are always $X$-visible. As copies of the same vertex of $H$ are connected via a clique in $K_L \times H$, we only need to consider the case of copies $u_i$ and $v_j$ of distinct vertices $u$ and $v$ of $H$ such that $(u_i,v_j)$ is not an edge in $K_L \times H$.

As $z_i \not \in S_i$ and $z_j \not \in S_j$, we have $z_i,z_j \not \in X$. Therefore, if $i= j$, both $u_i$ and $v_j=v_i$ are adjacent to $z_i$ and thus they are $X$-visible. 
For the case $i \neq j$ we split the proof into two subcases, according to whether $(u,v) \in E(H)$ or not. 

We consider the subcase $(u,v) \not \in E(H)$. First of all, we observe that $u_i$ and $v_j$ are at a distance of 3 in $K_L \times H$. As $K_L \times H$ contains the three edges $(u_i,z_i)$, $(z_i,z_j)$, and $(z_j,v_j)$ and $z_i,z_j \not \in X$, it follows that  $u_i$ and $v_j$ are $X$-visible.  

We consider the subcase $(u,v) \in E(H)$. First of all, we observe that $u_i$ and $v_j$ are a distance of 2 in $K_L \times H$. In particular, there are two shortest paths from $u_i$ to $v_j$: the one passing through $u_j$ and the one passing through $v_i$. As both $S_i$ and $S_j$ are copies of $S$, they are both independent set of layers $H_i$ and $H_j$, respectively. This implies that either $u_i,u_j \in X$ or $v_i,v_j \in X$. In either cases, $u_i$ and $v_j$ are $X$-visible.
\end{proof}

\begin{corollary}\label{cor:lb_mu_clique_graph}
Let $H$ be a graph with $N\geq 2$ vertices containing a  vertex $z$ of degree $N-1$. Let $L$ be a positive integer. Then $\mut(K_L \times H)\geq L \cdot \alpha(H)$.
\end{corollary}

We can show the inapproximability results on graphs of diameter 3.

\begin{theorem}\label{theo:inapx_diam_3}
For every $\tau \in \{\mu,\mud,\muo,\mut\}$ and for every constant $\varepsilon > 0$, the problem of computing a $\tau$-set of maximum cardinality on graphs $G$ with $n$ vertices and diameter of at least $3$ cannot be approximated within a factor of $n^{1/3-\varepsilon}$ for $\tau \in \{\mu,\mud\}$ and a factor of $n^{1/2-\varepsilon}$ for $\tau \in \{\muo,\mut\}$, unless $\mathsf{P} = \mathsf{NP}$.
\end{theorem}
\begin{proof}
First of all we observe that, for every constant $\delta>0$, the problem of computing a maximum independent set of a graph $H$ with $N\geq 2$ vertices and containing a vertex of degree $N-1$ is not approximable within a factor of $O(N^{1-\delta})$, unless $\mathsf{P} = \mathsf{NP}$. This is because given a graph $H'$ with $N-1$ vertices as an input instance of the \mis problem, we can construct a new graph $H$ by adding to $H'$ a new vertex $z$ of degree $N-1$ that is adjacent to all the vertices of $H'$. As $\alpha(H)=\alpha(H')$, any independent set of $H$ whose size approximates $\alpha(H)$ within a factor of $O(N^{1-\delta})$ can be used to compute an independent set of $H'$ whose size approximates $\alpha(H')$ within a factor of $O(N^{1-\delta})$ in polynomial time.

Given a graph $H$ with $N\geq 2$ vertices containing a vertex $z$ of degree $N-1$ as our input instance of the \mis problem, we build the graph $K_{L}\times H$ having $n=NL=N^{1+\beta}$ vertices by setting $L=N^\beta$, where $\beta=2$ if $\tau \in \{\mu,\mud\}$ and $\beta=1$ if $\tau \in \{\muo,\mut\}$. 

By~\Cref{lemma:is_graph_feasible_clique_graph}, $K_L\times H$ is a graph of diameter $3$. Let $\varepsilon = \delta/(\beta+1)$. We show that the existence of any polynomial-time algorithm that computes a $\tau$-set $X$ for $K_L \times H$ whose size approximates  $\tau(K_L \times H)$ within a factor of $n^{1/(\beta+1)-\varepsilon}$ would imply the existence of a polynomial-time algorithm that computes an independent set $S$ of $H$ whose size approximates $\alpha(H)$ within a factor of $O(N^{1-\delta})$, i.e., such that $|S|=\Omega\left(\frac{\alpha(H)}{N^{1-\delta}}\right)$.

For the sake of contradiction, assume that there is a polynomial-time algorithm that computes a $\tau$-set $X$ for $K_L \times H$ such that $|X|\geq \frac{\tau(K_{L}\times H)}{n^{1/(\beta +1)-\varepsilon}}$.
By~\Cref{cor:lb_mu_clique_graph}, we have that $\tau(K_{L}\times H)\geq \mut(K_{L}\times H)\geq  L\cdot \alpha(H)=N^{\beta} \cdot \alpha(H)$. Therefore, $|X|\geq \frac{N^{\beta}\cdot \alpha(H)}{N^{1-\delta}}$.

By~\Cref{lemma:property_cartesian_product_clique_graph}, given $X$ and $K_{L}\times H$, we can compute in polynomial time a subset $X'\subseteq X$ of size $|X'|\geq |X|-N^\beta$ such that, for every $i \in [L]$, $S_i:=X' \cap V(H_i)$ is an independent set of $H_i$. We now compute an independent set $S$ of $H$ as follows. 
If all $S_i$'s are empty sets, then $S=\{v\}$ where $v$ is any arbitrary vertex of $H$. If some $S_i$ are not empty, then let $i^* \in \arg\max \{|S_i|\mid  i \in [L]\}$ and let $S=\{v \mid v_{i^*} \in S_{i^*}\}$. As $|S_{i^*}|\geq |X'|/L$, we have that $|S|\geq |X'|/L$. By construction, we always return an independent set $S$ of size of at least $\max\{1,|X'|/L\}$. As a consequence, when $\alpha(H) < 2N^{1-\delta}$, $|S|$ already approximates $\alpha(H)$ within a factor of $O(N^{1-\delta})$. When $\alpha(H)\geq 2N^{1-\delta}$,  we have 
\begin{equation*}
|S| \geq \frac{|X'|}{L}= \frac{|X|-N^\beta}{N^\beta}\geq \frac{\frac{N^\beta\cdot \alpha(H)}{N^{1-\delta}}-N^\beta}{N^\beta}
    =\frac{\alpha(H)}{N^{1-\delta}}-1=\Omega\left(\frac{\alpha(H)}{N^{1-\delta}}\right).
    \tag*{\qedhere}
\end{equation*}
\end{proof}

\section{The general position problem and its relation with mutual visibility}
\label{sec:inaprox_gp}

In this section we consider a related visibility problem called \emph{general position} in which we want to compute a $\gp$-set $X \subseteq V(G)$ of a graph $G$, such that no shortest path in $G$ contains three distinct vertices from $X$.
We first show that the problem of finding a maximum-size $\gp$-set is not approximable within a factor of $n^{1-\varepsilon}$, unless $\mathsf{P}=\mathsf{NP}$, which was previously only known to be $\mathsf{NP}$-hard~\cite{manuel-2018}.
Then, we argue that $\frac{\mu(G)}{\gp(G)}$ can be as high as $\Omega( \frac{n}{\log^2 n} )$, and this is tight up to a $\log n$ factor, thus implying that the above inapproximability result does not carry over to the problem of computing a maximum $\mu$-set.

We start by stating the next  \Cref{lemma:indpedent_cliques_gp}, which is a consequence of Theorem~3.1 in~\cite{AnandCCKT19}.
\begin{definition}
    An \emph{independent clique} of a graph $G$ is a set $C \subseteq V(G)$ of vertices such that every connected component of $G[C]$ is a clique.
\end{definition}

\begin{lemma}[\cite{AnandCCKT19}]
    \label{lemma:indpedent_cliques_gp}
    Let $G$ be a graph with diameter $2$, and let $X \subseteq V(G)$. Then, $X$ is a $\gp$-set of $G$ if and only if $X$ is an independent clique of $G$.
\end{lemma}

We can now give our inapproximability result:
\begin{theorem}
    \label{thm:gp_inapx}
    There exists no polynomial-time approximation algorithm for general position achieving an approximation ratio of $n^{1-\varepsilon}$ for some constant $\varepsilon > 0$, unless $\mathsf{P} = \mathsf{NP}$.
\end{theorem}
\begin{proof}
Suppose towards a contradiction that there exists a polynomial-time $n^{1-\varepsilon}$-approximation algorithm $A$ for general position, where $\varepsilon < 1$ is some positive constant. 
The following \cliqueis problem was introduced in \cite{Eppstein10}: given a graph $G$, find a set of vertices of maximum cardinality that induces either an independent set or a clique of $G$.

Consider an instance $G$ of \cliqueis with a sufficiently large number $n$ of vertices, let $S^*$ be an optimal solution, and call $G'$ the graph obtained by adding a vertex $v$ adjacent to all vertices of $G$. 
We can obtain an approximation algorithm $A'$ for \cliqueis by running $A$ on $G'$ to obtain a $\gp$-set $X$, selecting a subset $\widetilde{X}$ of $X$ of size at least $\sqrt{|X|}$ that is either a clique or an independent set of $G'$, and returning $\widetilde{X} \setminus \{v\}$. As a technicality, we assume that $| \widetilde{X} | \ge 2$ (if this is not already the case, it suffices to select any arbitrary two vertices from $V(G')$).

By \Cref{lemma:indpedent_cliques_gp} we have that the size of an optimal $\gp$-set of $G'$ is at least $S^*$.
Moreover, given any $\gp$-set $C$ of $G'$, we can select a subset $C' \subseteq C$ of size at least $\sqrt{|C|}$ such that $C'$ is either an independent set or a clique of $G'$.\footnote{By \Cref{lemma:indpedent_cliques_gp}, $C$ is an independent clique of $G$. If $G[C]$ contains at least $\sqrt{|C|}$ connected components, we select an independent set by choosing on vertex from each component. Otherwise $G[C]$ contains at least one connected component of size $\sqrt{|C|}$, and this component is a clique.}

Therefore $|X| \ge \frac{|S^*|}{n^{1-\varepsilon}}$ and $\widetilde{X} \ge \sqrt{|X|} \ge \left( \frac{|S^*|}{n^{1-\varepsilon}} \right)^{1/2}$.
The approximation ratio achieved by $A'$ is upper bounded by:
\[
    \frac{|S^*|}{|\widetilde{X}|-1} \le 
    \frac{2 |S^*|}{|\widetilde{X}|} \le
    \frac{2 |S^*| \cdot \sqrt{ n^{1-\varepsilon}}}{\sqrt{|S^*|}} 
    = 2\sqrt{|S^*|} \cdot \sqrt{ n^{1-\varepsilon} }
    \le 2\sqrt{n} \cdot \sqrt{ n^{1-\varepsilon} }
    = 2n^{1 - \frac{\varepsilon}{2}},
\]
which contradicts the fact that \cliqueis is not approximable in polynomial-time within a factor $n^{1-\delta}$ for any positive constant $\delta < 1$. We conclude that, if $\mathsf{P} \neq \mathsf{NP}$, algorithm $A$ cannot exist.
\end{proof}

In view of the previous inapproximability result, one could hope to also obtain an inapproximability result for the problem of computing a maximum $\mu$-set for graphs having diameter $2$ by relating it with the problem of computing a maximum $\gp$-set on the same graph.
Indeed, defining $\beta(n)$ as a function that satisfies $\mu(G) \le \beta(n) \cdot \gp(G)$ for any graph $G$ having diameter $2$ and $n$ vertices, we would immediately have that the maximum $\mu$-set problem is not approximable within an approximation factor of $\Omega\left( \frac{n^{1-\varepsilon}}{\beta(n)} \right)$, for any constant $\varepsilon>0$, unless $\mathsf{P}=\mathsf{NP}$.

Unfortunately, the following lemma shows that no meaningful choice of $\beta(n)$ exists.
\begin{lemma}
    \label{lemma:beta_lb}
    Any function $\beta(n)$ such that  $\mu(G) \le \beta(n) \cdot \gp(G)$ for any graph $G$ of diameter $2$.
    is such that $\beta(n) = \Omega\left(\frac{n}{\log^2 n}\right)$.
\end{lemma}
\begin{proof}
    Let $R(k, k)$ be the minimum number of vertices such that any graph with $n \ge R(k,k)$ vertices contains a clique of $k$ vertices or an independent set of size $k$ (possibly both).
    It is known from Ramsey theory \cite{erdos1935combinatorial,erdos1947some} that:
    \begin{equation}
        \label{eq:ramsey}
        (1+o(1)) \cdot \frac{k}{e\sqrt{2}} 2^{k/2} \le R(k,k) \le (1 + o(1)) \cdot \frac{4^{k-1}}{\sqrt{\pi k}}.
    \end{equation}

    For a given $k \ge 2$, let $G'$ be a graph with $R(k,k)-1$ vertices  that contains neither an independent set of size $k$ nor a clique of size $k$ and let $G$ be the graph obtained by adding a new vertex adjacent to all vertices of $G'$. 
    
    We have that $G$ has $n=R(k,k)$ vertices, diameter $2$, and contains neither an independent set of size $k+1$ nor a clique of size $k+1$.
    From the first inequality of \Cref{eq:ramsey} we know that $n = \Omega(k 2^{k/2})$, which implies $k = O\left(\log \left( \frac{n}{\log n} \right) \right) = O(\log n)$.

    Thus, any set of vertices of $G$ that induces a subgraph consisting of independent cliques has size at most $k^2$. From \Cref{lemma:indpedent_cliques_gp}, we conclude that $\gp(G) = O(k^2)$ while $\mu(G)  \ge n - 1 = \Omega(n)$ since $V(G')$ is a $\mu$-set of $G$. This implies $\beta(n) =\Omega( \frac{n}{k^2} ) = \Omega(\frac{n}{\log^2 n})$.
\end{proof}

Next lemma shows that the bound of \Cref{lemma:beta_lb} is almost tight:
\begin{lemma}
    There exists a function $\beta(n) = O\left(\frac{n}{\log n}\right)$ such that  $\mu(G) \le \beta(n) \cdot \gp(G)$ for any graph $G$ of diameter $2$.
\end{lemma}
\begin{proof}
Let $G$ be any graph with diameter $2$ and $n$ vertices, and let $k$ be the largest integer such that $R(k,k) \le n$.
Then $n < R(k+1, k+1) = O(4^k/\sqrt{k}) = O(4^k)$ from the second inequality in \Cref{eq:ramsey}, which implies $k = \Omega(\log n)$.
Since $n \ge R(k, k)$, $G$ contains a clique of size $k$ or an independent set of size $k$. By \Cref{lemma:indpedent_cliques_gp}, $\gp(G) \ge k$ and we have:
\[
\mu(G) \le n = \frac{n}{k} \cdot k \le \frac{n}{k} \cdot \gp(G) = O\left(\frac{n}{\log n}\right) \cdot \gp(G),
\]
thus showing that $\beta(n) = O\left(\frac{n}{\log n}\right)$.
\end{proof}

\section{Open problems}\label{sec:open}
For the mutual-visibility number we presented a polynomial-time algorithm that finds a $\mu$-set with size $\Omega\left( \sqrt{n/ \overline{D}} \right)$, where $\overline{D}$ is the average distance in $G$. It would be interesting to study better algorithms to close the gap with the inapproximability result shown in \Cref{sec:InapVis} for the same problem. 
Of course, providing good approximation algorithms for finding the maximum values of the other invariants is also a challenging problem.

Given the inapproximability results for the problem of finding a maximum $\tau$-set, for $\tau \in \{\mu, \mud, \muo, \mut\}$, it would be relevant to introduce a relaxed version of the mutual visibility. Given a set of vertices $X$, we  could say that two vertices are in visibility if there exists a path connecting each pair of vertices without a further vertex in $X$ such that the length of the path is at most $\sigma$ times the distance between the two vertices, 
where $\sigma\geq 1$ is a fixed constant. 
Then, it would be interesting studying the computational complexity of the corresponding visibility problems, as well as the properties of the relative invariants from a graph theoretical point of view. 

In this setting, we observe that the inapproximability results of \Cref{tau-apx} and \Cref{mut-diam2} carry over to the case in which $1 \leq \sigma <3/2$ as they hold for the class of diameter-2 graphs. Similarly, the inapproximability results of \Cref{theo:inapx_diam_3} on graphs of diameter 3 also hold when $1 \leq \sigma < 4/3$. Moreover, we observe that the reduction provided in \Cref{tau-apx}, that we slightly modify by removing the vertex $y$ from $G$, results in a graph $G$ of diameter 2 in which the problem of computing a $\mu$-set of maximum size is still $\mathsf{APX}$-hard for every $\sigma \geq 1$.

\bibliographystyle{plain}
\bibliography{references}

\begin{thebibliography}{10}

\bibitem{hardnessCubic}
Paola Alimonti and Viggo Kann.
\newblock Hardness of approximating problems on cubic graphs.
\newblock In {\em Proceedings of {CIAC} '97}, volume 1203 of {\em LNCS}, pages
  288--298. Springer, 1997.

\bibitem{AnandCCKT19}
Bijo~S. Anand, S.~V.~Ullas Chandran, Manoj Changat, Sandi Klavzar, and
  Elias~John Thomas.
\newblock Characterization of general position sets and its applications to
  cographs and bipartite graphs.
\newblock {\em Applied mathematics and computation}, 359:84--89, 2019.

\bibitem{bhagat20}
Subhash Bhagat.
\newblock Optimum algorithm for the mutual visibility problem.
\newblock In M.~Sohel Rahman, Kunihiko Sadakane, and Wing{-}Kin Sung, editors,
  {\em Proceedings of {WALCOM:} 2020}, volume 12049 of {\em LNCS}, pages
  31--42. Springer, 2020.

\bibitem{Bresar}
Boštjan Brešar and Ismael~G. Yero.
\newblock Lower (total) mutual-visibility number in graphs.
\newblock {\em Applied Mathematics and Computation}, 465:128411, 2024.

\bibitem{Bujtas}
Csilla Bujtás, Sandi Klavžar, and Jing Tian.
\newblock Total mutual-visibility in hamming graphs.
\newblock {\em arXiv}, 2023.

\bibitem{CaroT91}
Yair Caro and Zsolt Tuza.
\newblock Improved lower bounds on \emph{k}-independence.
\newblock {\em Journal of Graph Theory}, 15(1):99--107, 1991.

\bibitem{ullas-2016}
Ullas S.~V. Chandran and G.~Jaya Parthasarathy.
\newblock The geodesic irredundant sets in graphs.
\newblock {\em International journal of mathematical combinatorics},
  4:135--143, 2016.

\bibitem{CiceroneFSN23}
Serafino Cicerone, Alessia {Di Fonso}, Gabriele {Di Stefano}, and Alfredo
  Navarra.
\newblock The geodesic mutual visibility problem for oblivious robots: the case
  of trees.
\newblock In {\em Proceedings of {ICDCN} 2023}, pages 150--159. {ACM}, 2023.

\bibitem{CDDN23-PMC}
Serafino Cicerone, Alessia {Di Fonso}, Gabriele {Di Stefano}, and Alfredo
  Navarra.
\newblock The geodesic mutual visibility problem: Oblivious robots on grids and
  trees.
\newblock {\em Pervasive and Mobile Computing}, 95:101842, 2023.

\bibitem{lagos23}
Serafino Cicerone and Gabriele {Di Stefano}.
\newblock Mutual-visibility in distance-hereditary graphs: a linear-time
  algorithm.
\newblock In {\em Proceedings of {LAGOS} 2023}, volume 223, pages 104--111.
  Elsevier, 2023.

\bibitem{variety-2023}
Serafino Cicerone, Gabriele {Di Stefano}, Lara Drozdek, Jaka Hedzet, Sandi
  Klavzar, and Ismael~G. Yero.
\newblock Variety of mutual-visibility problems in graphs.
\newblock {\em Theoretical Computer Science}, 974:114096, 2023.

\bibitem{CiceroneDK23}
Serafino Cicerone, Gabriele {Di Stefano}, and Sandi Klavžar.
\newblock On the mutual visibility in cartesian products and triangle-free
  graphs.
\newblock {\em Applied Mathematics and Computation}, 438:Paper 127619, 2023.

\bibitem{EJB24}
Serafino Cicerone, Gabriele {Di Stefano}, Sandi Klavžar, and Ismael~G. Yero.
\newblock Mutual-visibility problems on graphs of diameter two.
\newblock {\em European Journal of Combinatorics}, 120:103995, 2024.

\bibitem{CiceroneFSN23-SSS}
Serafino Cicerone, Alessia~Di Fonso, Gabriele {Di Stefano}, and Alfredo
  Navarra.
\newblock Time-optimal geodesic mutual visibility of robots on grids within
  minimum area.
\newblock In {\em Proceedings of {SSS} 2023}, volume 14310 of {\em Lecture
  Notes in Computer Science}, pages 385--399. Springer, 2023.

\bibitem{CFSNP24}
Serafino Cicerone, Alessia~Di Fonso, Gabriele~Di Stefano, Alfredo Navarra, and
  Francesco Piselli.
\newblock Mutual visibility in hypercube-like graphs.
\newblock In {\em Proceedings of {SIROCCO} 2024}, volume 14662 of {\em LNCS},
  pages 192--207. Springer, 2024.

\bibitem{LunaFCPSV17}
Giuseppe~Antonio {Di Luna}, Paola Flocchini, Sruti~Gan Chaudhuri, Federico
  Poloni, Nicola Santoro, and Giovanni Viglietta.
\newblock Mutual visibility by luminous robots without collisions.
\newblock {\em Information and Computation}, 254:392--418, 2017.

\bibitem{DiStefano22}
Gabriele {Di Stefano}.
\newblock Mutual visibility in graphs.
\newblock {\em Applied Mathematics and Computation}, 419:~Paper 126850, 2022.

\bibitem{Dudeney17}
H.~E. Dudeney.
\newblock {\em Amusements in mathematics}.
\newblock Nelson, Edinburgh, 1917.

\bibitem{Eppstein10}
David Eppstein.
\newblock Paired approximation problems and incompatible inapproximabilities.
\newblock In {\em Proceedings of {SODA} 2010}, pages 1076--1086. {SIAM}, 2010.

\bibitem{erdos1947some}
Paul Erd{\"o}s.
\newblock Some remarks on the theory of graphs.
\newblock {\em Bulletin of the American Mathematical Society}, 53:292--294,
  1947.

\bibitem{erdos1935combinatorial}
Paul Erd{\"o}s and George Szekeres.
\newblock A combinatorial problem in geometry.
\newblock {\em Compositio mathematica}, 2:463--470, 1935.

\bibitem{klavzar-2021}
Sandi Klavzar, Bal{\'a}zs Patk{\'o}s, Gregor Rus, and Ismael~G. Yero.
\newblock On general position sets in cartesian products.
\newblock {\em Results in Mathematics}, 76:123, 2021.

\bibitem{kuziak-2023}
Dorota Kuziak and Juan~A. Rodr\'{\i}guez-Vel\'{a}zquez.
\newblock Total mutual-visibility in graphs with emphasis on lexicographic and
  cartesian products.
\newblock {\em Bulletin of the Malaysian Mathematical Sciences Society},
  46:197, 2023.

\bibitem{manuel-2018}
Paul Manuel and Sandi Klav{\v z}ar.
\newblock A general position problem in graph theory.
\newblock {\em Bulletin of the Australian Mathematical Society},
  98(2):177–187, 2018.

\bibitem{ManuelK18}
Paul~D. Manuel and Sandi Klavzar.
\newblock The graph theory general position problem on some interconnection
  networks.
\newblock {\em Fundamenta Informaticae}, 163(4):339--350, 2018.

\bibitem{PoudelAS21}
Pavan Poudel, Aisha Aljohani, and Gokarna Sharma.
\newblock Fault-tolerant complete visibility for asynchronous robots with
  lights under one-axis agreement.
\newblock {\em Theoretical Computer Science}, 850:116--134, 2021.

\bibitem{ShachnaiS01}
Hadas Shachnai and Aravind Srinivasan.
\newblock Finding large independent sets of hypergraphs in parallel.
\newblock In Arnold~L. Rosenberg, editor, {\em Proceedings of {SPAA} 2001},
  pages 163--168. {ACM}, 2001.

\bibitem{SharmaVT21}
Gokarna Sharma, Ramachandran Vaidyanathan, and Jerry~L. Trahan.
\newblock Optimal randomized complete visibility on a grid for asynchronous
  robots with lights.
\newblock {\em International Journal of Networking and Computing},
  11(1):50--77, 2021.

\bibitem{tian-2023+}
Jing Tian and Sandi Klavžar.
\newblock Graphs with total mutual-visibility number zero and total
  mutual-visibility in cartesian products.
\newblock {\em Discussiones Mathematicae Graph Theory}, 2024.

\bibitem{tian-2021}
Jing Tian and Kexiang Xu.
\newblock The general position number of cartesian products involving a factor
  with small diameter.
\newblock {\em Applied Mathematics and Computation}, 403:126206, 2021.

\bibitem{Zuckerman07}
David Zuckerman.
\newblock Linear degree extractors and the inapproximability of max clique and
  chromatic number.
\newblock {\em Theory of Computing}, 3(1):103--128, 2007.

\end{thebibliography}

\end{document}